\pdfoutput=1
\RequirePackage{fix-cm}
\documentclass[twoside,english,round, colon, dvipsnames,12pt]{article}
\usepackage[T1]{fontenc}
\usepackage[latin9]{inputenc}
\usepackage{geometry}
\usepackage{color}
\usepackage{babel}
\usepackage{bm}
\usepackage{amsmath}
\usepackage{amsthm}
\usepackage{amssymb}
\usepackage{graphicx}
\usepackage[authoryear]{natbib}
\usepackage[unicode=true,
 bookmarks=false,
 breaklinks=true,pdfborder={0 0 0},pdfborderstyle={},backref=false,colorlinks=true]
 {hyperref}
\hypersetup{
  citecolor=red!20!blue}

\makeatletter
  \theoremstyle{plain}
  \newtheorem{lem}{\protect\lemmaname}
\theoremstyle{plain}
\newtheorem{thm}{\protect\theoremname}
  \theoremstyle{definition}
  \newtheorem{defn}{\protect\definitionname}

\usepackage{microtype,lmodern}
\usepackage[bb=boondox]{mathalfa}
\RequirePackage{xcolor}

\@ifundefined{showcaptionsetup}{}{%
 \PassOptionsToPackage{caption=false}{subfig}}
\usepackage{subfig}
\makeatother

  \providecommand{\definitionname}{Definition}
  \providecommand{\lemmaname}{Lemma}
\providecommand{\theoremname}{Theorem}

\begin{document}
\global\long\def\mb#1{\hm{#1}}
\global\long\def\mbb#1{\mathbb{#1}}
\global\long\def\mc#1{\mathcal{#1}}
\global\long\def\mcc#1{\mathscr{#1}}
\global\long\def\mr#1{\mathrm{#1}}
\global\long\def\msf#1{\mathsf{#1}}
\global\long\def\E{\mbb E}
\global\long\def\P{\mbb P}
\global\long\def\var{\mr{var}}
\global\long\def\T{\msf T}
\global\long\def\d{\mr d}
\global\long\def\tr{\msf{tr}}
\global\long\def\F{\mr F}
\global\long\def\argmax{\operatorname*{argmax}}
\global\long\def\argmin{\operatorname*{argmin}}
\global\long\def\defeq{\stackrel{\textup{\tiny def}}{=}}
\global\long\def\bbone{\mbb 1}

\date{}
\title{Phase Retrieval Meets Statistical Learning Theory:\\
A Flexible Convex Relaxation}

\author{{Sohail~Bahmani\thanks{The authors were supported in part by ONR grant N00014-11-1-0459,
NSF grants CCF-1415498 and CCF-1422540, and the Packard Foundation.} \qquad Justin~Romberg\footnotemark[1]}\\[1ex]
\normalsize School of Electrical and Computer Engineering\\
\normalsize Georgia Instititute of Technology\\
\normalsize \texttt{\{sohail.bahmani,jrom\}@ece.gatech.edu}
}
\maketitle

\begin{abstract}
We propose a flexible convex relaxation for the phase retrieval problem
that operates in the natural domain of the signal. Therefore, we avoid
the prohibitive computational cost associated with ``lifting'' and
semidefinite programming (SDP) in methods such as \emph{PhaseLift}
and compete with recently developed non-convex techniques for phase
retrieval. We relax the quadratic equations for phaseless measurements
to inequality constraints each of which representing a symmetric ``slab''.
Through a simple convex program, our proposed estimator finds an extreme
point of the intersection of these slabs that is best aligned with
a given \emph{anchor vector}. We characterize geometric conditions
that certify success of the proposed estimator. Furthermore, using
classic results in statistical learning theory, we show that for random
measurements the geometric certificates hold with high probability
at an optimal sample complexity. Phase transition of our estimator
is evaluated through simulations. Our numerical experiments also suggest
that the proposed method can solve phase retrieval problems with coded
diffraction measurements as well.
\end{abstract}

\section{\label{sec:intro}Introduction}

Let $\mb x_{\star}\in\mbb C^{N}$ be a signal that we would like to
recover from noisy phaseless measurements 
\begin{align}
b_{i} & =\left|\mb a_{i}^{*}\mb x_{\star}\right|^{2}+\xi_{i} & i=1,2,\dotsc,M,\label{eq:noisy}
\end{align}
where the measurement vectors $\mb a_{i}\in\mbb C^{N}$ are given.
To solve this \emph{phase retrieval} problem with provable accuracy,
different methods that rely on semidefinite relaxation have been proposed
previously \citep[e.g.,][]{Candes_PhaseLift_2013,Candes_Solving_2014,Waldspurger_Phase_2015}.
While these methods are guaranteed to produce an accurate solution
in polynomial time, they are not scalable due to the use of semidefinite
programming (SDP). This drawback of SDP-based methods has motivated
development of alternative non-convex methods that operate in the
natural domain of the signal and exhibit better scalability \citep[e.g.,][]{Netrapalli_Phase_2013,Candes_Phase_2014}.
In this paper we follow a completely different approach and propose
a convex relaxation of the phase retrieval problem that not only produces
accurate solutions but also is scalable. Compared to the non-convex
phase retrieval methods our approach inherits the flexibility of convex
optimization both in analysis and application.

\begin{figure}
\noindent\centering\includegraphics[width=0.75\columnwidth]{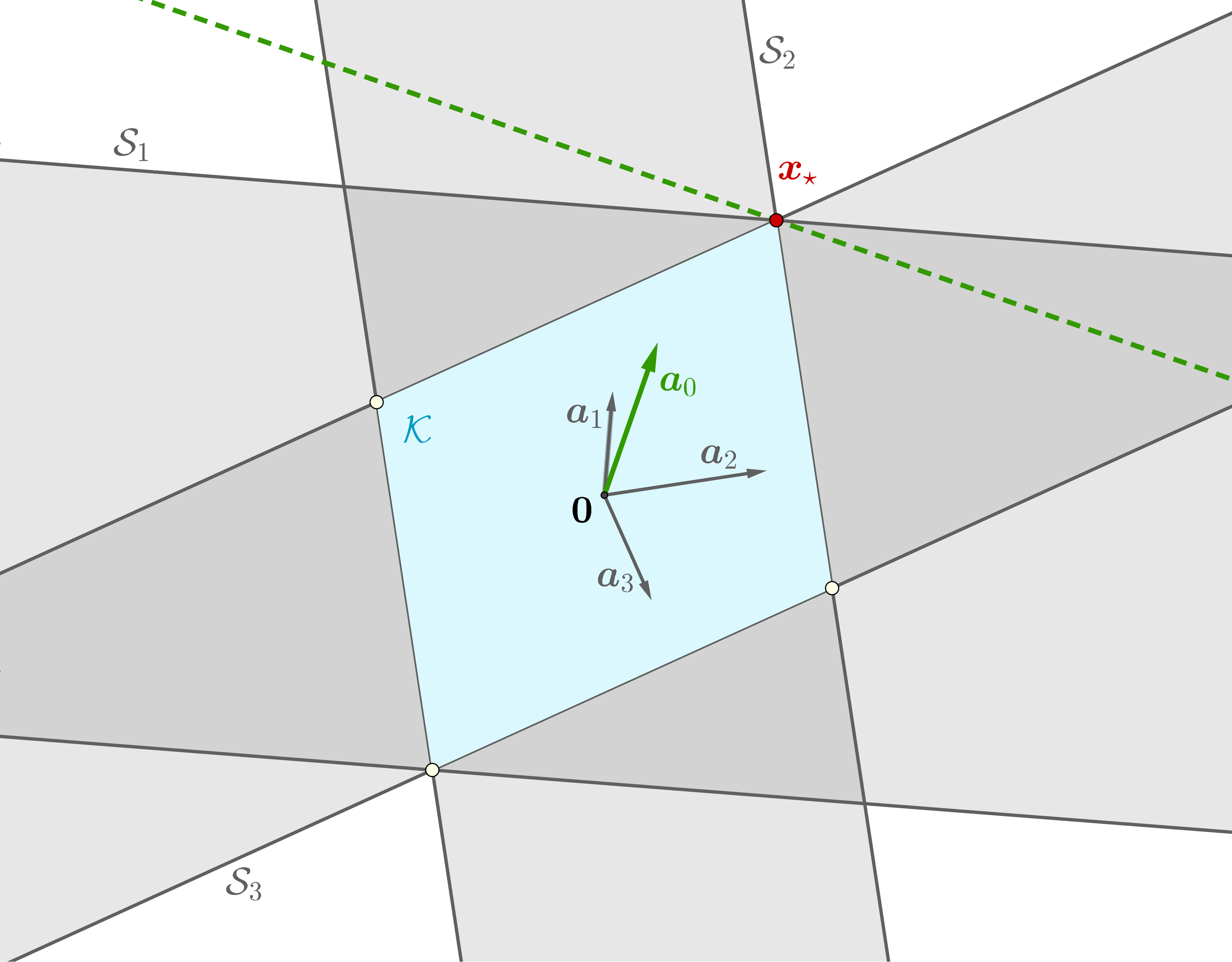}\caption{\label{fig:slabs}A two-dimensional illustration of slabs intersecting
at $\protect\mb x_{\star}$}
\end{figure}
The geometric idea at the core of our proposed method is the following.
Relaxing each measurement equation in (\ref{eq:noisy}) to an inequality
$\left|\mb a_{i}^{*}\mb x_{\star}\right|^{2}\le b_{i}$ creates a
symmetric slab $\mc S_{i}$ of feasible solutions as illustrated in
Figure \ref{fig:slabs}. Collectively, these slabs describe a ``complex
polytope'' $\mc K$ of feasible solutions. In the noiseless regime
(i.e., $\xi_{i}=0$ for all $i$), the target signal $\mb x_{\star}$
would be one of the extreme points of $\mc K$. To distinguish $\mb x_{\star}$
among all of the extreme points, our idea is to find a hyperplane
tangent to $\mc K$ at $\mb x_{\star}$. The crucial ingredient in
this approach is an ``anchor'' vector $\mb a_{0}\in\mbb C^{N}\backslash\left\{ \mb 0\right\} $
that acts as the normal for the desired tangent hyperplane and it
is required to have a non-vanishing correlation with $\mb x_{\star}$
in the sense that 
\begin{align}
\left|\mb a_{0}^{*}\mb x_{\star}\right| & \ge\delta\left\lVert \mb a_{0}\right\rVert _{2}\left\lVert \mb x_{\star}\right\rVert _{2},\label{eq:well-cor}
\end{align}
 for some absolute constant $\delta\in\left(0,1\right)$. The above
geometric intuition is explained in more detail in Section \ref{ssec:geometry}.
While our main result simply assumes that the anchor vector is given
by an oracle, which may use the existing measurements, we discuss
in Section \ref{ssec:a_0} some realistic scenarios where a valid
anchor vector exists or can be computed.

We assume that the noise is non-negative (i.e., $\xi_{i}\ge0$) and
we have $\left\lVert \mb{\xi}\right\rVert _{\infty}\le\eta^{-1}\left\lVert \mb x_{\star}\right\rVert _{2}^{2}$
for some constant $\eta>0$. Note that the non-negativity of the noise
can be dropped at the cost of a slight reduction in the effective
signal-to-noise ratio. In particular, one can add the noise upperbound
(i.e., $\eta^{-1}\left\lVert \mb x_{\star}\right\rVert _{2}^{2}$)
to each measurement to ensure the non-negativity. Throughout we treat
$\mbb C^{N}$ as an inner-product space over $\mbb R$ equipped with
the symmetric inner-product 
\[
\langle\cdot,\cdot\rangle:\left(\mb x_{1},\mb x_{2}\right)\mapsto\mr{Re}\left(\mb x_{1}^{*}\mb x_{2}\right).
\]
Clearly, in this setting $\mbb C^{N}$ will be a $2N$-dimensional
vector space.  

With these assumptions in place, we propose the solution to the convex
program\footnote{In the real case, (\ref{eq:maxCorr}) reduces to a linear program.}
\begin{align}
\max_{\mb x}\, & \langle\mb a_{0},\mb x\rangle\label{eq:maxCorr}\\
\text{subject to } & \left|\mb a_{i}^{*}\mb x\right|^{2}\le b_{i} & 1\le i\le M,\nonumber 
\end{align}
 as a computationally efficient estimator for $\mb x_{\star}$. Of
course, the points equal to $\mb x$ up to a global phase, namely,
\[
\mbb T\mb x\defeq\left\{ \omega\mb x\,:\,\left|\omega\right|=1\right\} ,
\]
 yield the same phaseless measurements. Therefore, the goal is merely
to estimate a point in $\mbb T\mb x_{\star}$ accurately from the
phaseless measurements (\ref{eq:noisy}).

In Lemma \ref{lem:suff-cond-det}, below in Section \ref{sec:proofs},
we establish a geometric condition that is sufficient to guarantee
accurate estimation of $\mb x_{\star}$ via the convex program (\ref{eq:maxCorr}).

The sufficient condition given by Lemma \ref{lem:suff-cond-det} can
be interpreted in terms of (non-)existence of a particularly constrained
halfspace that includes all of the points $\mb a_{i}\mb a_{i}^{*}\mb x_{\star}$.
For random measurement vectors $\mb a_{i}$, this interpretation resembles
the model and theory of \emph{linear classifiers} studied in \emph{statistical
learning theory,} albeit in an unusual regime\emph{.} Borrowing classic
results from this area (summarized in Appendix \ref{sec:SLT}), we
show that with high probability (\ref{eq:maxCorr}) produces an accurate
estimate of $\mb x_{\star}$.

Specifically, in our main result, Theorem \ref{thm:main-thm} in Section
\ref{sec:proofs}, we show that drawing
\[
M\overset{\delta}{\gtrsim}N+\log\frac{1}{\varepsilon},
\]
i.i.d. random measurements, with the hidden constant factor on the
right-hand side depending on $\delta$, would suffice for the conditions
of Lemma \ref{lem:suff-cond-det} to hold with probability $\ge1-\varepsilon$.
Consequently, solution $\widehat{\mb x}$ of (\ref{eq:maxCorr}) would
obey 
\[
\left\lVert \widehat{\mb x}-\mb x_{\star}\right\rVert _{2}\lesssim\eta^{-1}\left\lVert \mb x_{\star}\right\rVert _{2}.
\]

\subsection{\label{ssec:a_0}Choosing the anchor vector}

Our approach critically depends on the choice of the anchor vector
$\mb a_{0}$ that obeys an inequality of the form (\ref{eq:well-cor}).
Below we discuss two interesting scenarios where such a vector would
be accessible.

\paragraph{Non-negative signals:}

Perhaps the simplest scenario is when the target signal $\mb x_{\star}$
is known to be real and non-negative. In usual imaging modalities
these model assumptions are realistic as natural images are typically
represented by pixel intensities. For these types of signals we can
choose $\mb a_{0}=\frac{1}{\sqrt{N}}\mb 1$ for which we obtain $\left|\mb a_{0}^{*}\mb x_{\star}\right|=\left\lVert \mb x_{\star}\right\rVert _{1}/\sqrt{N}$.
Then, for (\ref{eq:well-cor}) to hold it suffices that $\left\lVert \mb x_{\star}\right\rVert _{1}\ge\delta\sqrt{N}\left\lVert \mb x_{\star}\right\rVert _{2}$
for some absolute constant $\delta\in(0,1)$. In particular, we need
$\mb x_{\star}$ to have at least $\delta^{2}N$ non-zero entries.

\paragraph{Random measurements:}

A more interesting scenario is when we can construct the vector $\mb a_{0}$
from the (random) measurements. An effective strategy is to set $\mb a_{0}$
to be the principal eigenvector of the matrix $\mb{\varSigma}=\frac{1}{M}\sum_{i=1}^{M}b_{i}\mb a_{i}\mb a_{i}^{*}$.
The principal eigenvector of $\mb{\varSigma}$ and its ``truncated''
variants have been used previously for initialization of the \emph{Wirtinger
Flow} algorithm \citep{Candes_Phase_2014} and its refined versions
\citep{Chen-Solving-2015,Zhang-Provable-2016}. For example, the following
result is shown in \citet[Section VII.H]{Candes_Phase_2014}.
\begin{lem}[\citet{Candes_Phase_2014}]
\label{lem:a0-PCA} For $1\le i\le M$ let $b_{i}$ be the phaseless
measurements obtained from i.i.d. vectors $\mb a_{i}\sim\mr{Normal}(\mb 0,\frac{1}{2}\mb I)+\imath\mr{Normal}(\mb 0,\frac{1}{2}\mb I)$
and no noise. If $M\overset{\delta}{\gtrsim}N\log N$ and $\mb a_{0}$
is the principal eigenvector of 
\[
\mb{\varSigma}=\frac{1}{M}\sum_{i=1}^{M}b_{i}\mb a_{i}\mb a_{i}^{*},
\]
 then (\ref{eq:well-cor}) holds with probability $\ge1-O(N^{-2})$.
\end{lem}
 While Lemma \ref{lem:a0-PCA} can be refined or extended in various
ways, we do not pursue these paths in this paper.

\subsection{Related work}

There is a large body of research on phase retrieval addressing various
aspect of the problem (see \citep{Jaganathan-Phase-2015} and references
therein). However, we focus only on the relevant results mostly developed
in recent years. Perhaps, among the most important developments are
\emph{PhaseLift }and similar methods that cast the phase retrieval
problem as a particular semidefinite program \citep{Candes_PhaseLift_2013,Candes_Solving_2014,Waldspurger_Phase_2015}.
The main idea used by \citet{Candes_PhaseLift_2013} and \citet{Candes_Solving_2014}
is that by \emph{lifting} the unknown signal using the transformation
$\mb x\mb x^{*}\mapsto\mb X$, the (noisy) phaseless measurements
(\ref{eq:noisy}) that are quadratic in $\mb x_{\star}$ can be converted
to linear measurements of the rank-one positive semidefinite matrix
$\mb X_{\star}=\mb x_{\star}\mb x_{\star}^{*}$. With this observation,
these SDP-based methods aim to solve the corresponding linear equations
using the trace-norm to induce the rank-one structure in the solution.
Inspired by the well-known convex relaxation of Max-Cut problem, \emph{PhaseCut
}method \citep{Waldspurger_Phase_2015} considers the measurement
phases as the unknown variables and applies a similar lifting transform
to formulate a different semidefinite relaxation for phase retrieval.
While these SDP-based methods are shown to produce accurate estimates
of $\mb X_{\star}$ at optimal sample complexity for certain random
measurement models, they become computationally prohibitive in medium-
to large-scale problems where SDP is practically inefficient.

More recently, there has been a growing interest in non-convex iterative
methods for phase retrieval \citep[see e.g.,][]{Netrapalli_Phase_2013,Candes_Phase_2014,Schniter-Compressive-2015,Chen-Solving-2015,Zhang-Provable-2016,Wang_Solving_2016a,Sun-Geometric-2016}.
These methods typically operate in the natural space of the signal
and thus do not suffer the drawbacks of the SDP-based methods. With
a specific initialization \citet{Netrapalli_Phase_2013} establish
some accuracy guarantees for a variant of the classic methods by \citet{Gerchberg-Practical-1972,Fienup-Phase-1982}
that iteratively update the estimate assuming the measurements' phase
match that of the previous iterate. The established sample complexity
is (nearly) optimal in the dimension of the target signal, but it
does not vary gracefully with the prescribed precision. Phase retrieval
via the \emph{Wirtinger Flow} (WF), a non-convex gradient descent method
at core, is proposed by \citet{Candes_Phase_2014}. It is shown that
for random measurements that have Normal distribution or certain \emph{
coded diffraction patterns}, with an appropriate initialization the
WF iterates exhibit the linear rate of convergence to
the target signal. More recent work on the WF method
introduce better initialization by excluding the outlier measurements
and achieve the optimal sample complexity \citep{Chen-Solving-2015,Zhang-Provable-2016}. The WF class of algorithms and our proposed method both achieve optimal sample complexity (up to the constant factor) and have low computational cost. However, the WF methods need careful tuning of a step size parameter and their convergence analysis often relies on Gaussian measurements. This is partly because establishing robustness of non-convex methods generally requires stronger conditions. Our method provably works for a broader set of measurement distributions, has no tuning parameters, and can be implemented in various convex optimization software.

Shortly after a draft of this manuscript was first posted online, a few independent papers proposed and analyzed the same method and its variants. \citet{Goldstein-PhaseMax-2016}, who dubbed (\ref{eq:maxCorr})\emph{
PhaseMax}, obtained sharper constants in the sample complexity by assuming
a stronger condition that the anchor is independent of the measurements
in their analysis. Alternative proofs and variations that rely on
matrix concentration inequalities appeared later in \citep{Hand-Elementary-2016,Hand-Corruption-2016,Hand-Compressed-2016}.
Another distinctive feature of our analysis compared to the mentioned
results is that it is less sensitive to measurement distribution as
it relies on VC\textendash type bounds.

\subsection{Variations and Extensions}

In this section we discuss several different ways to extend the proposed
method that we leave for future research. While the core geometric
idea still applies, some modifications of our theoretical arguments
would be necessary to analyze these extensions.

The gross noise model considered in this paper can be pessimistic
in scenarios where we have random noise or deterministic noise with
a different type of bound. In these scenarios, augmenting the estimator
by a noise regularization term could result in accuracy bounds that
gracefully vary with the considered noise.

Another interesting extension to the proposed method, is to adapt
the current theory to the case of blockwise independent measurements
as in coded diffraction imaging. Our numerical experiments in Section
(\ref{sec:simulations}) suggest that the proposed method still performs
well with these structured measurements. Nevertheless, to extend the
analysis we may need to revise the current simple arguments based
on Vapnik-Chervonenkis theory using more sophisticated tools from
the theory of empirical processes.

Finally, we believe that our proposed method is flexible in the sense
that it allows to incorporate a structural properties of the signal
relatively easily. In particular, it would be interesting to analyze
a variant of the proposed estimator that induces sparsity through
\mbox{$\ell_{1}$-norm} regularization.

\section{\label{sec:simulations}Numerical experiments}

We evaluated the performance of our proposed method on synthetic data
with the target signal $\mb x_{\star}\sim\mr{Normal}(\mb 0,\frac{1}{2}\mb I)+\imath\mr{Normal}(\mb 0,\frac{1}{2}\mb I)$
and measurements $\mb a_{i}\overset{\text{\tiny i.i.d.}}{\sim}\mr{Normal}(\mb 0,\frac{1}{2}\mb I)+\imath\mr{Normal}(\mb 0,\frac{1}{2}\mb I)$
all having $N=500$ coordinates. The noisy measurements follow (\ref{eq:noisy})
with the uniform noise $\xi_{i}\overset{\text{\tiny i.i.d.}}{\sim}\mr{Uniform}(\left[0,\eta^{-1}\right])$
in one experiment and the Gaussian noise $\xi_{i}\overset{\text{\tiny i.i.d.}}{\sim}\mr{Normal}(0,\sigma^{2})$
in the other. For the latter noise model we replaced any negative
$b_{i}$ by $b_{i}=0$ to avoid negative measurements and defined
the input signal-to-noise ratio as $\mr{SNR}\defeq10\log_{10}\frac{\left\lVert \mb x_{\star}\right\rVert _{2}^{4}}{\sigma^{2}}$.
The vector $\mb a_{0}$ is constructed as in initialization of the
Wirtinger Flow mentioned in Lemma \ref{lem:a0-PCA} through $50$
iterations of the power method. We implemented the convex program
(\ref{eq:maxCorr}) by TFOCS \citep{Becker_Templates_2011} with smoothing
parameter $\mu=2\times10^{-3}$ and at most $500$ iterations. Figure
\ref{fig:phase-trans} illustrates the $0.9$-quantile and median
of the relative error $\min_{\phi\in[0,2\pi)}\left\lVert \widehat{\mb x}-e^{\imath\phi}\mb x_{\star}\right\rVert _{2}/\left\lVert \mb x_{\star}\right\rVert _{2}$
observed over $100$ trials of our algorithm for different sampling
ratios $\frac{M}{N}$ between $2$ and $17$. The plots also show
the effect of different levels of noise on the relative error for
both of the considered noise models.
\begin{figure}
\noindent\centering\subfloat[Uniform noise model]{\includegraphics[width=0.5\columnwidth]{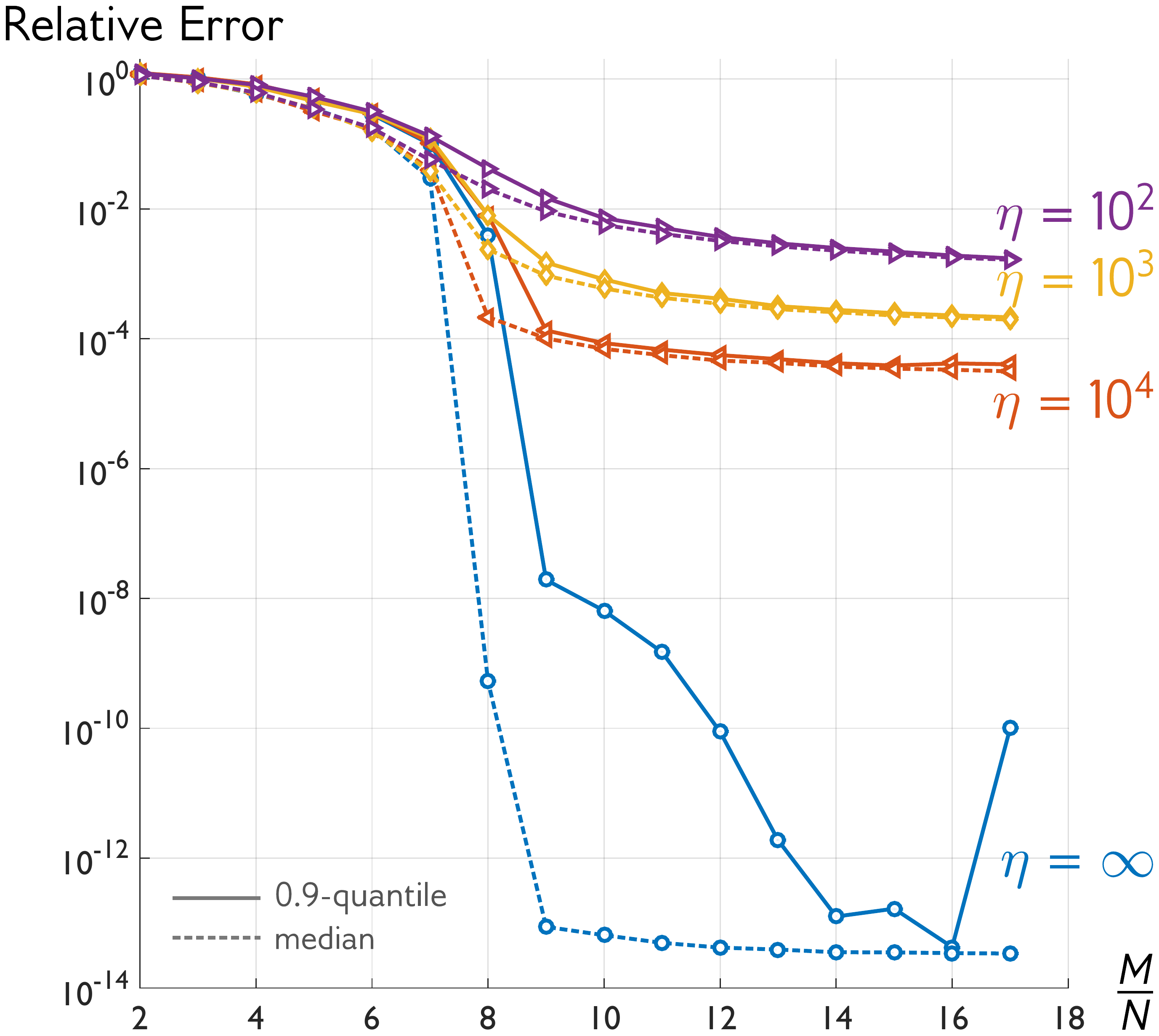}}
\subfloat[Gaussian noise model]{\includegraphics[width=0.5\columnwidth]{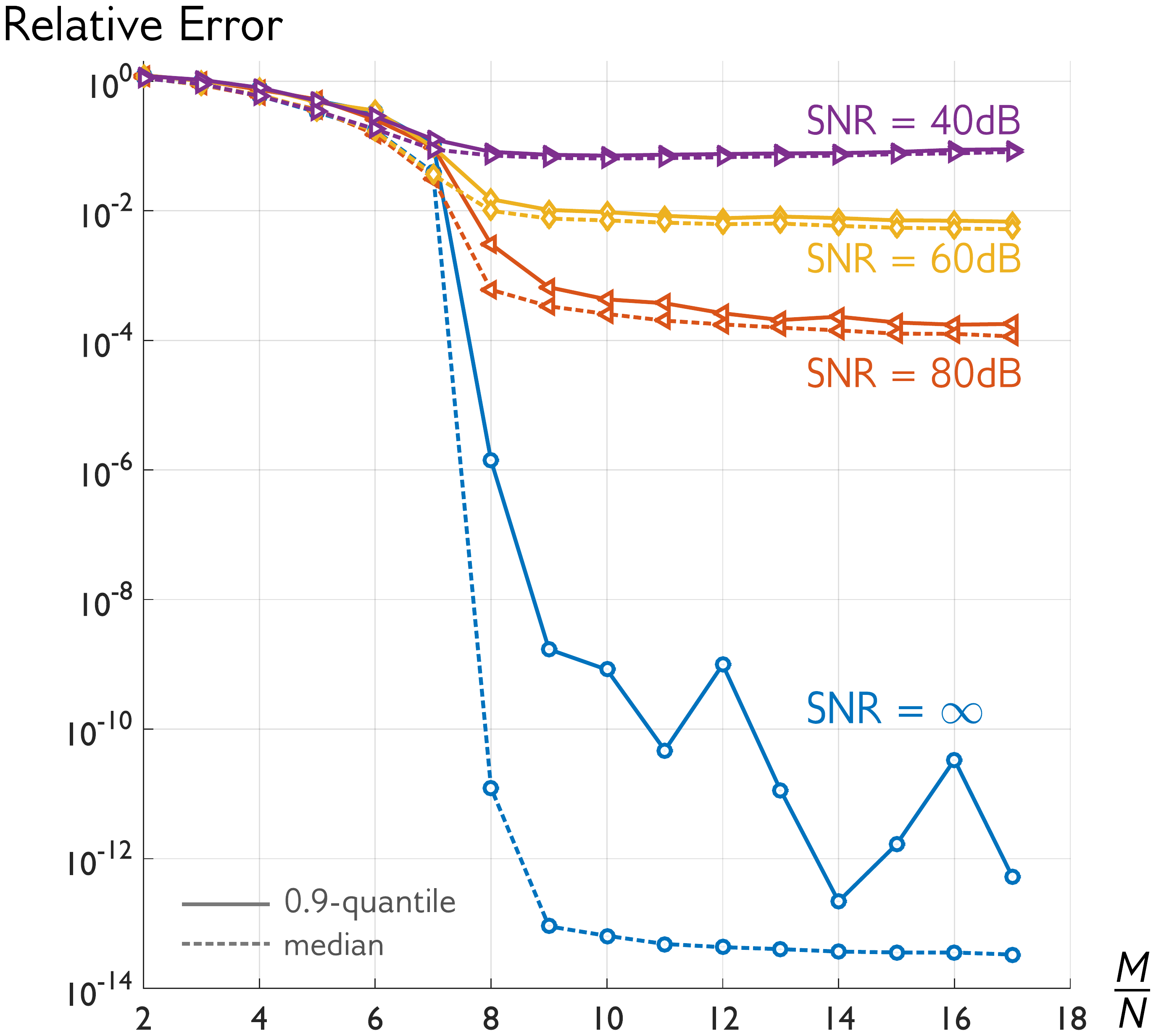}}\caption{\label{fig:phase-trans}Empirical relative error of the proposed method
at different sampling ratios (i.e., $\frac{M}{N}$) and noise levels
with $N=500$}
\end{figure}
We also evaluated our method using noiseless measurements with coded
diffraction patterns as described by \citet{Candes_Phase-CDP_2015}.
Specifically, with indices $i=\left(k,\ell\right)$ for $1\le k\le N$
and $1\le\ell\le L$, we used measurements of the form $\mb a_{i}=\mb f_{k}\circ\mb{\phi}_{\ell}$
which is the pointwise product of the $k$-th discrete Fourier basis
vector (i.e., $\mb f_{k}$) and a random modulation pattern with i.i.d.
symmetric Bernoulli entries (i.e., $\mb{\phi}_{\ell}$). The target
signal is an $N=960\times1280\approx1.2\times10^{6}$ pixel image
of a Persian Leopard.\footnote{\texttt{}%
\begin{minipage}[t]{0.75\columnwidth}%
Available online at:\\
\scriptsize \texttt{https://upload.wikimedia.org/wikipedia/commons/thumb/7/7d/Persian\_Leopard\_sitting.jpg/1280px-Persian\_Leopard\_sitting.jpg}%
\end{minipage}} We used $L=20$ independent coded diffraction patterns $\left\{ \mb{\phi}_{\ell}\right\} _{1\le\ell\le L}$.
Therefore, the total number of (scalar) measurements is $M=LN\approx2.5\times10^{7}$.
Similar to the first simulation, the vector $\mb a_{0}$ is constructed
as the (approximate) principal eigenvector of $\frac{1}{M}\sum_{i}b_{i}\mb a_{i}\mb a_{i}^{*}$
through 50 iterations of the power method. The convex program is also
solved using TFOCS, but this time with smoothing parameter $\mu=10^{-6}$
and restricting the total number of forward and adjoint coded diffraction
operator to $500$. The recovered image is depicted in Figure \ref{fig:Persian-Leopard}
which has a relative error of about $8.2\times10^{-8}$.

\begin{figure}
\noindent\centering\includegraphics[width=0.75\columnwidth]{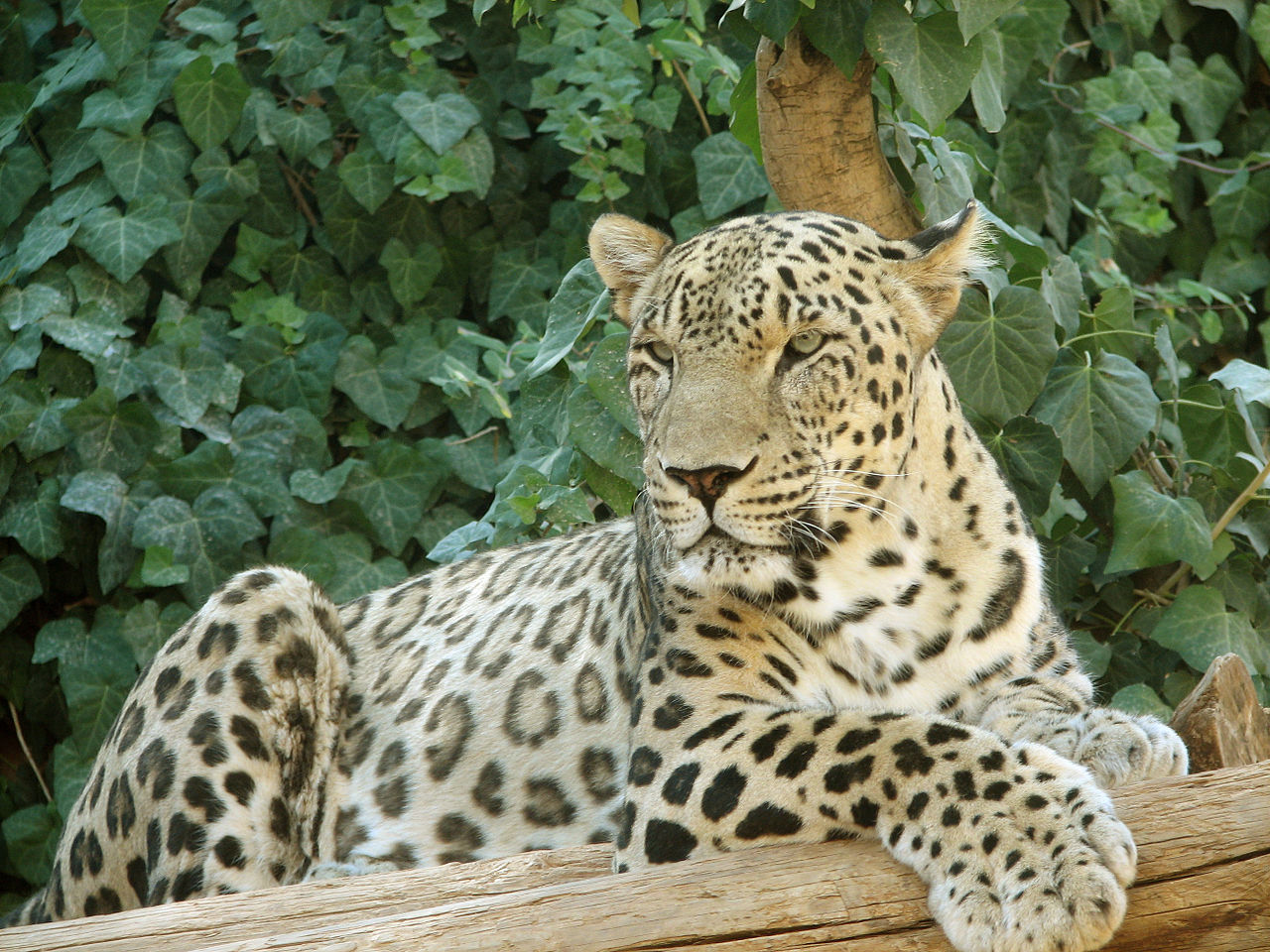}\caption{\label{fig:Persian-Leopard}Persian Leopard at $960\times1280$ resolution.
Relative error is $8.2\times10^{-8}$.}
\end{figure}

\section{\label{sec:proofs}Theoretical Analysis}

In this section we provide the precise statement of the our results
and their proofs. For the sake of simplicity in notation and derivation,
but without loss of generality, we make the following assumptions.
We assume that $\mb a_{0}^{*}\mb x_{\star}$ is a positive real number
since any point in $\mbb T\mb x_{\star}$ is a valid target. Furthermore,
we assume that $\mb x_{\star}$ is unit-norm (i.e., $\left\lVert \mb x_{\star}\right\rVert _{2}=1$)
and thus the bound on the noise reduces to $\left\lVert \mb{\xi}\right\rVert _{\infty}\le\eta^{-1}$.
We first establish, in Lemma \ref{lem:suff-cond-det}, a geometric
condition for success of phase retrieval through (\ref{eq:maxCorr}).
Then we use this lemma to prove our main result for random measurements
in Theorem \ref{thm:main-thm}. We also rely on tools from statistical
learning theory that are outlined in Appendix \ref{sec:SLT}.

\subsection{\label{ssec:geometry}Geometry of intersecting slabs}

To understand the geometry of (\ref{eq:maxCorr}) it is worthwhile
to first consider the noiseless scenario. The feasible set is the
intersection of the sets
\[
\mc S{}_{i}=\left\{ \mb x\in\mbb C^{N}\,:\,\left|\mb a_{i}^{*}\mb x\right|^{2}\le b_{i}\right\} 
\]
 corresponding to the pairs $\left(\mb a_{i},b_{i}\right)$ for $i=1,2,\dotsm,M$.
The sets $\mc S_{i}$ are effectively symmetric ``complex slabs''.
Denote their intersection by\emph{ 
\[
\mc K=\bigcap_{i=1}^{M}\mc S_{i}.
\]
} In (\ref{eq:maxCorr}) the objective function is linear, thus its
maximizer is an extreme point of the convex constraint set $\mc K$.
Clearly, $\mb x_{\star}$ as well as any other point in $\mbb T\mb x_{\star}$
are extreme points of $\mc K$. However, $\mc K$ typically has other
extreme points that are not equivalent to $\mb x_{\star}$. Intuitively,
using the non-vanishing correlation of $\mb a_{0}$ with $\mb x_{\star}$,
the convex program (\ref{eq:maxCorr}) is effectively eliminating
the superfluous extreme points of $\mc K$. The geometric interpretation
is that the hyperplane normal to $\mb a_{0}$ that passes through
$\mb x_{\star}$ is also tangent to $\mc K$, as Figure \ref{fig:slabs}
suggests. It is not difficult to show that an analogous interpretation
from the dual point of view is that $\mb a_{0}$ is in the interior
of the conical hull $\mr{cone}\left\{ \mb a_{i}\mb a_{i}^{*}\mb x_{\star}\right\} _{1\le i\le N}$.

More generally, with noisy measurements, $\mc K$ is still a symmetric
complex polytope that is convex and includes $\mbb T\mb x_{\star}$
due to non-negativity of the noise\emph{.} We would like to find conditions
that guarantee that the solution to (\ref{eq:maxCorr}) is close to
$\mb x_{\star}$. More specifically, we would like to show that if
$\widehat{\mb x}=\mb x_{\star}+\mb h$ is any solution to (\ref{eq:maxCorr})
and $t>0$ is some constant, then with $\left\lVert \mb h\right\rVert _{2}>\left(t\eta\right)^{-1}$
the inequalities
\begin{align*}
\langle\mb a_{0},\mb h\rangle & \ge0\\
\left|\mb a_{i}^{*}\left(\mb x_{\star}+\mb h\right)\right|^{2} & \le\left|\mb a_{i}^{*}\mb x_{\star}\right|^{2}+\xi_{i} & 1\le i\le M,
\end{align*}
cannot hold simultaneously. The following lemma provides the desired
sufficient condition.
\begin{lem}
\label{lem:suff-cond-det} Let
\begin{align}
\mc R_{\delta} & =\left\{ \mb h\in\mbb C^{N}\,:\,\left\lVert \mb h-\left(\mb x_{\star}^{*}\mb h\right)\mb x_{\star}\right\rVert _{2}\ge\delta\left|\mr{Im}\left(\mb x_{\star}^{*}\mb h\right)\right|\right\} \,,\label{eq:R_delta}
\end{align}
and $\epsilon\ge0$ be some constant. If every vector $\mb h\in\mc R_{\delta}$
with $\left\lVert \mb h\right\rVert _{2}>\epsilon$ violates at least
one of the inequalities
\begin{align}
\begin{aligned}\langle\mb a_{0},\mb h\rangle & \ge0\\
\langle\mb a_{i}\mb a_{i}^{*}\mb x_{\star},\mb h\rangle & \le\frac{1}{2}\eta^{-1} & 1\le i\le M,
\end{aligned}
\label{eq:primal-opt}
\end{align}
 then any solution $\widehat{\mb x}$ to (\ref{eq:maxCorr}) obeys
\[
\left\lVert \widehat{\mb x}-\mb x_{\star}\right\rVert _{2}\le\epsilon.
\]
\end{lem}
\begin{proof}
It suffices to show that $\mb h=\widehat{\mb x}-\mb x_{\star}$ obeys
(\ref{eq:primal-opt}) and it belongs to $\mc R_{\delta}$. Given
that 
\begin{align*}
\xi_{i}\ge\left|\mb a_{i}^{*}\left(\mb x_{\star}+\mb h\right)\right|^{2}-\left|\mb a_{i}^{*}\mb x_{\star}\right|^{2} & =\left|\mb a_{i}^{*}\mb h\right|^{2}+2\langle\mb a_{i}\mb a_{i}^{*}\mb x_{\star},\mb h\rangle\ge2\langle\mb a_{i}\mb a_{i}^{*}\mb x_{\star},\mb h\rangle
\end{align*}
 and $\xi_{i}\le\eta^{-1}$, we have $\langle\mb a_{i}\mb a_{i}^{*}\mb x_{\star},\mb h\rangle\le\frac{1}{2}\eta^{-1}$.
Feasibility of $\mb x_{\star}$ also guarantees that $\langle\mb a_{0},\mb h\rangle\ge0$.
Therefore, we have shown that $\mb h$ satisfies (\ref{eq:primal-opt}).

The constraints of (\ref{eq:maxCorr}) are invariant under a global
change of phase (i.e., the action of $\mbb T$). It easily follows
that the solution $\widehat{\mb x}$ to (\ref{eq:maxCorr}) should
obey $\mr{Im}\left(\mb a_{0}^{*}\widehat{\mb x}\right)=0$. Therefore,
we have $\mr{Im}\left(\mb a_{0}^{*}\mb h\right)=0$ as we assumed
$\alpha=\mb a_{0}^{*}\mb x_{\star}\in\mbb R$. The same assumption
also implies that $\mb a_{0}=\alpha\mb x_{\star}+\mb a_{0\perp}$
for $\mb a_{0\perp}=\left(\mb I-\mb x_{\star}\mb x_{\star}^{*}\right)\mb a_{0}$
which clearly obeys $\mb x_{\star}^{*}\mb a_{0\perp}=0$. Thus, using
triangle inequality and the bound (\ref{eq:well-cor}) we obtain
\begin{align*}
0=\left|\mr{Im}\left(\mb a_{0}^{*}\mb h\right)\right| & =\left|\alpha\mr{Im}\left(\mb x_{\star}^{*}\mb h\right)+\mr{Im}\left(\mb a_{0\perp}^{*}\mb h\right)\right|\\
 & \ge\alpha\left|\mr{Im}\left(\mb x_{\star}^{*}\mb h\right)\right|-\left|\mr{Im}\left(\mb a_{0\perp}^{*}\mb h\right)\right|\\
 & \ge\delta\left\lVert \mb a_{0}\right\rVert _{2}\left|\mr{Im}\left(\mb x_{\star}^{*}\mb h\right)\right|-\left\lVert \mb h_{\perp}\right\rVert _{2}\left\lVert \mb a_{0}\right\rVert _{2},
\end{align*}
 where $\mb h_{\perp}=\left(\mb I-\mb x_{\star}\mb x_{\star}^{*}\right)\mb h=\mb h-\left(\mb x_{\star}^{*}\mb h\right)\mb x_{\star}$.
The above inequality completes the proof as it is equivalent to $\mb h\in\mc R_{\delta}$.
\end{proof}

\subsection{Guarantees for random measurements}

In this section we will show that if the vectors $\mb a_{i}$ for
$1\le i\le M$ are drawn from a random distribution and (\ref{eq:well-cor})
holds for a sufficiently large constant $\delta$, then with high
probability (\ref{eq:maxCorr}) produces an accurate estimate of $\mb x_{\star}$.
Our strategy is to show that for a sufficiently large $M$ the sufficient
condition provided in Lemma \ref{lem:suff-cond-det} holds with high
probability.

For $\delta\in(0,1)$ let $\mc C_{\delta}$ be the convex cone given
by 
\[
\mc C_{\delta}=\left\{ \mb y\in\mbb C^{N}\,:\,\mb x_{\star}^{*}\mb y\ge\delta\left\lVert \mb y\right\rVert _{2}\right\} \,,
\]
 where $\mb x_{\star}^{*}\mb y$ is implicitly assumed to be a real
number. The \emph{polar cone }of a set $\mc C$ is defined as 
\[
\mc C^{^{\circ}}\defeq\left\{ \mb z\,:\,\langle\mb z,\mb y\rangle\le0\textup{ for all }\mb y\in\mc C\right\} .
\]
 It is easy to verify that the polar cone of $\mc C_{\delta}$ is
\begin{align*}
\mc C_{\delta}^{^{\circ}} & =\left\lbrace\mb z\in\mbb C^{N}:\delta\langle\mb x_{\star},\mb z\rangle\!\le\!-\sqrt{1\!-\!\delta^{2}}\sqrt{\left\lVert \mb z\right\rVert _{2}^{2}-\left|\mb x_{\star}^{*}\mb z\right|^{2}}\right\rbrace.
\end{align*}
 Since $\mb a_{0}\in\mc C_{\delta}$ by assumption, it follows that
for every $\mb h\in\mc C_{\delta}^{^{\circ}}$ we have $\langle\mb a_{0},\mb h\rangle\le0$.
Therefore, the inequality $\langle\mb a_{0},\mb h\rangle\ge0$ can
hold only for vectors $\boldsymbol{z}$ in the closure of the complement
of $\mc C_{\delta}^{^{\circ}}$ which we denote by 
\begin{align}
\mc C_{\delta}^{'} & =\left\{ \mb z\hspace{-0.5ex}\in\hspace{-0.5ex}\mbb C^{N}\,:\,\delta\langle\mb x_{\star},\mb z\rangle\hspace{-0.5ex}\ge\hspace{-0.5ex}-\sqrt{1\!-\!\delta^{2}}\sqrt{\left\lVert \mb z\right\rVert _{2}^{2}-\left|\mb x_{\star}^{*}\mb z\right|^{2}}\right\} .\label{eq:C'_delta}
\end{align}
 A typical positioning of $\mb a_{0}$ and $\mb a_{i}\mb a_{i}^{*}\mb x_{\star}$
needed to guarantee unique recovery is illustrated in Figure \ref{fig:cones}.

\begin{figure}
\noindent
\centering

\includegraphics[width=.75\columnwidth]{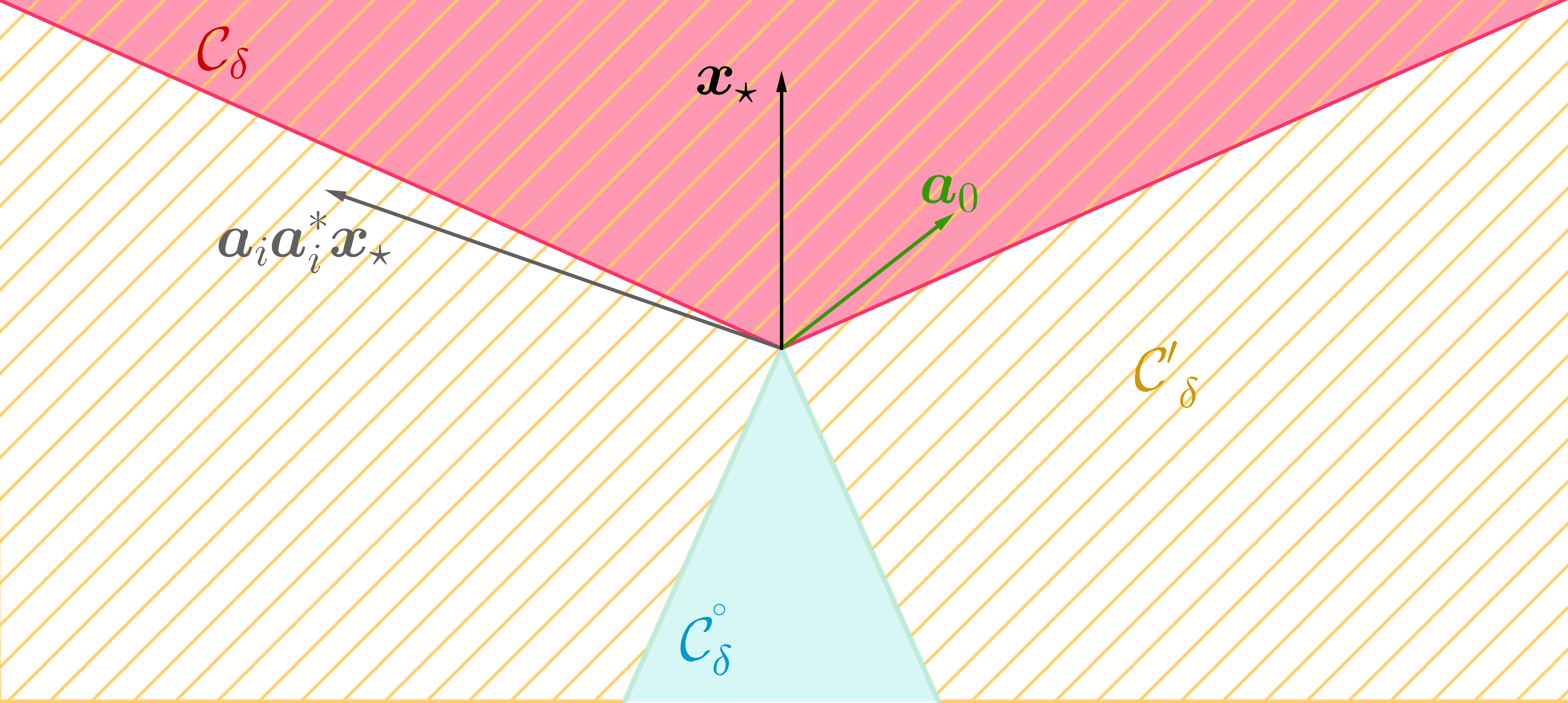}\caption{\label{fig:cones}Relative positioning of $\protect\mb a_{0}$ and
$\protect\mb a_{i}\protect\mb a_{i}^{*}\protect\mb x_{\star}$ with
respect to the cones $\protect\mc C_{\delta}$ and $\protect\mc C'_{\delta}$}
\end{figure}

\begin{thm}
\label{thm:main-thm} Suppose that noisy phaseless measurements of
a unit vector $\mb x_{\star}$ as in (\ref{eq:noisy}) are given under
bounded non-negative noise $\xi_{1},\xi_{2},\dotsc,\xi_{M}\in\left[0,\eta^{-1}\right]$.
Let $\mb a_{0}$ be an anchor vector obeying (\ref{eq:well-cor})
for some constant $\delta\in(0,1)$ and define $\mc R_{\delta}$ and
$\mc C'_{\delta}$ respectively as in (\ref{eq:R_delta}) and (\ref{eq:C'_delta}).
Furthermore, given a constant $t>0$, suppose that for $1\le i\le M$
the measurement vectors $\mb a_{i}$ are i.i.d. copies of a random
variable $\mb a\in\mbb C^{N}$ that obeys 
\[
\inf_{\substack{\mb h\in\mc C'_{\delta}\cap\mc R_{\delta}\\
\left\lVert \mb h\right\rVert _{2}>\left(t\eta\right)^{-1}
}
}\P\left(\langle\mb a\mb a^{*}\mb x_{\star},\mb h\rangle>\frac{1}{2}\eta^{-1}\right)\ge p_{\min}(\delta,t),
\]
 for a constant $p_{\min}(\delta,t)\in\left(0,1\right)$ depending
only on $\delta$ and $t$.\footnote{Clearly, the best $p_{\min}\left(\delta,t\right)$ decreases as $t$
increases.} For any $\varepsilon>0$, if we have
\[
M\overset{p_{\min}\left(\delta,t\right)}{\gtrsim}N+\log\frac{1}{\varepsilon},
\]
with the hidden constant factor inversely related to $p_{\min}(\delta,t)$,
then with probability $\ge1-\varepsilon$ the estimate $\widehat{\mb x}$
obtained through (\ref{eq:maxCorr}) obeys 
\[
\left\lVert \widehat{\mb x}-\mb x_{\star}\right\rVert _{2}\le\left(t\eta\right)^{-1}.
\]
\end{thm}
\begin{proof}
Let $\mb h=\widehat{\mb x}-\mb x_{\star}$. It suffices to show that
for any $\mb h\in\mc C'_{\delta}\cap\mc R_{\delta}$ with $\left\lVert \mb h\right\rVert _{2}>\left(t\eta\right)^{-1}$
there exists at least one $1\le i\le M$ such that $\langle\mb a_{i}\mb a_{i}^{*}\mb x_{\star},\mb h\rangle>\frac{1}{2}\eta^{-1}$.
Specifically, we would like to show that with high probability $\sum_{i=1}^{M}\bbone\left(\langle\mb a_{i}\mb a_{i}^{*}\mb x_{\star},\mb h\rangle>\frac{1}{2}\eta^{-1}\right)>0$.
Denote the empirical probability of $\langle\mb a\mb a^{*}\mb x_{\star},\mb h\rangle>\frac{1}{2}\eta^{-1}$
by 
\[
\widehat{p}_{M}(\mb h)=\frac{1}{M}\sum_{i=1}^{M}\bbone\left(\langle\mb a_{i}\mb a_{i}^{*}\mb x_{\star},\mb h\rangle>\frac{1}{2}\eta^{-1}\right),
\]
which is an approximation of the true probability of the event denoted
by
\begin{align*}
p(\mb h) & =\E\bbone\left(\langle\mb a\mb a^{*}\mb x_{\star},\mb h\rangle>\frac{1}{2}\eta^{-1}\right)=\P\left(\langle\mb a\mb a^{*}\mb x_{\star},\mb h\rangle>\frac{1}{2}\eta^{-1}\right).
\end{align*}
Considering the set of binary functions $\mc F=\left\{ \mb z\mapsto\bbone\left(\langle\mb z,\mb h\rangle>\frac{1}{2}\eta^{-1}\right)\,:\,\mb h\in\mc C'_{\delta}\cap\mc R_{\delta}\text{ and }\left\lVert \mb h\right\rVert _{2}>\left(t\eta\right)^{-1}\right\} $
whose shatter coefficient is denoted by $s(\mc F,M)$, a direct application
of Theorem \ref{thm:VC} in Appendix \ref{sec:SLT} shows that
\[
\sup_{\substack{\mb h\in\mc C'_{\delta}\cap\mc R_{\delta}\\
\left\lVert \mb h\right\rVert _{2}>\left(t\eta\right)^{-1}
}
}\left|\widehat{p}_{M}(\mb h)-p(\mb h)\right|\le\sqrt{\frac{8\log\frac{8\,s(\mc F,M)}{\varepsilon}}{M}}
\]
with probability $\ge1-\varepsilon$. Since $\mc F$ is a subset of
$\mc H$ the set of indicators of all half-spaces (with a common offset),
it has a smaller VC\textendash dimension than $\mc H$. Moreover,
it is well-known\textemdash as a direct implication of Radon's theorem
\citep[see e.g.,][]{Matousek-Lectures-2002}\textemdash that the VC\textendash dimension
of half-spaces indicators is no more than the ambient dimension. In
particular, we have $\dim_{\mr{VC}}(\mc F)\le\dim_{\mr{VC}}(\mc H)\le2N$
as our domain is effectively a $2N$-dimensional real vector space.
Therefore, invoking Lemma \ref{lem:Sauer's} below we obtain
\[
\sup_{\substack{\mb h\in\mc C'_{\delta}\cap\mc R_{\delta}\\
\left\lVert \mb h\right\rVert _{2}>\left(t\eta\right)^{-1}
}
}\left|\widehat{p}_{M}(\mb h)-p(\mb h)\right|\le\sqrt{\frac{16N\log\frac{eM}{2N}+8\log\frac{8}{\varepsilon}}{M}}.
\]
 Now, because $p(\mb h)\ge p_{\min}(\delta,t)$ for all $\mb h\in\mc C'_{\delta}\cap\mc R_{\delta}$
with $\left\lVert \mb h\right\rVert _{2}>\left(t\eta\right)^{-1}$,
the above inequality implies that
\[
\inf_{\substack{\mb h\in C'_{\delta}\cap\mc R_{\delta}\\
\left\lVert \mb h\right\rVert _{2}>\left(t\eta\right)^{-1}
}
}\hspace{-0.5ex}\widehat{p}_{M}(\mb h)\ge p_{\min}(\delta,t)\hspace{-0.5ex}-\hspace{-0.5ex}\sqrt{\frac{16N\log\frac{eM}{2N}\hspace{-0.5ex}+\hspace{-0.5ex}8\log\frac{8}{\varepsilon}}{M}}.
\]
If $M=\frac{8}{p_{\min}^{2}(\delta,t)}\left(c\cdot2N+2\log\frac{8}{\varepsilon}\right)$,
then we have 
\begin{align*}
\log\frac{eM}{2N}+\frac{\log\frac{8}{\varepsilon}}{2N} & =\log\frac{8e}{p_{\min}^{2}\left(\delta,t\right)}\,+\log\left(c+\frac{\log\frac{8}{\varepsilon}}{N}\right)+\frac{\log\frac{8}{\varepsilon}}{2N}\\
 & \le\log\frac{8e}{p_{\min}^{2}\left(\delta,t\right)}\,+\frac{c}{2}+\frac{\log\frac{8}{\varepsilon}}{2N}-1+\log2+\frac{\log\frac{8}{\varepsilon}}{2N}\\
 & <\log\frac{8e}{p_{\min}^{2}\left(\delta,t\right)}\,+\frac{c}{2}+\frac{\log\frac{8}{\varepsilon}}{N},
\end{align*}
 where we used the inequality $\log u-\log2=\log\frac{u}{2}\le\frac{u}{2}-1$
in the second line. Setting $c=2\log\frac{8e}{p_{\min}^{2}\left(\delta,t\right)}$,
it follows that 
\begin{multline*}
\frac{16N\log\frac{eM}{2N}+8\log\frac{8}{\varepsilon}}{M}=\hspace{-0.5ex}\frac{16N}{M}\hspace{-0.5ex}\left(\hspace{-0.5ex}\log\frac{eM}{2N}+\frac{\log\frac{8}{\varepsilon}}{2N}\right)\\
<\hspace{-0.5ex}\frac{16N}{M}\hspace{-0.5ex}\left(\hspace{-0.5ex}2\log\frac{8e}{p_{\min}^{2}\left(\delta,t\right)}\,+\frac{\log\frac{8}{\varepsilon}}{N}\right)=p_{\min}^{2}\left(\delta,t\right)
\end{multline*}
 and thus we can guarantee that 
\[
\inf_{\substack{\mb h\in\mc C'_{\delta}\cap\mc R_{\delta}\\
\left\lVert \mb h\right\rVert _{2}>\left(t\eta\right)^{-1}
}
}\widehat{p}_{M}(\mb h)>p_{\min}(\delta,t)-p_{\min}(\delta,t)=0.
\]
 This immediately implies that for $M\overset{\delta,t}{\gtrsim}N+\log\frac{1}{\varepsilon}$
we have
\[
\inf_{\substack{\mb h\in\mc C'_{\delta}\cap\mc R_{\delta}\\
\left\lVert \mb h\right\rVert _{2}>\left(t\eta\right)^{-1}
}
}\!\sum_{i=1}^{M}\bbone(\langle\mb a_{i}\mb a_{i}^{*}\mb x_{\star},\mb h\rangle>\frac{1}{2}\eta^{-1})=M\widehat{p}_{M}(\mb h)>0\,,
\]
 as desired.
\end{proof}
We can consider the case of measurements with normal distribution
as a concrete example. To apply the Theorem \ref{thm:main-thm}, it
suffices to quantify the constant $p_{\min}(\delta,t)$ which can
be achieved through Lemma \ref{lem:normal-bound} below.
\begin{lem}
\label{lem:normal-bound}If $\mb a\sim\mr{Normal}(\mb 0,\frac{1}{2}\mb I)+\imath\mr{Normal}(\mb 0,\frac{1}{2}\mb I)$
and $\mb x_{\star}$ is a unit vector, then for every $\mb h\in\mc C'_{\delta}\cap\mc R_{\delta}$
with $\left\lVert \mb h\right\rVert _{2}>\left(t\eta\right)^{-1}$
we have 
\[
\P\left(\langle\mb a\mb a^{*}\mb x_{\star},\mb h\rangle>\frac{1}{2}\eta^{-1}\right)\ge\left(\frac{1}{2}-\frac{\sqrt{1-\delta^{2}}}{2}\right)e^{-2\sqrt{2}\delta^{-2}t}.
\]
\end{lem}
\begin{proof}
We can decompose $\mb h$ as $\mb h=\left(\mb x_{\star}^{*}\mb h\right)\mb x_{\star}+\mb h_{\perp}$,
where $\mb x_{\star}^{*}\mb h_{\perp}=\mb 0$. Therefore, we have
$\left\langle \mb a\mb a^{*}\mb x_{\star},\mb h\right\rangle =\left\langle \mb x_{\star},\mb h\right\rangle \left|\mb a^{*}\mb x_{\star}\right|^{2}+\mr{Re}\left(\overline{\mb a^{*}\mb x_{\star}}\mb a^{*}\mb h_{\perp}\right)$.
Using the facts that $\mb x_{\star}^{*}\mb h_{\perp}=\mb 0$ and $\mb a\sim\mr{Normal}(\mb 0,\frac{1}{2}\mb I)+\imath\mr{Normal}(\mb 0,\frac{1}{2}\mb I)$
it is straightforward to show that $\left\langle \mb x_{\star},\mb h\right\rangle \left|\mb a^{*}\mb x_{\star}\right|^{2}+\mr{Re}\left(\overline{\mb a^{*}\mb x_{\star}}\mb a^{*}\mb h_{\perp}\right)$
has the same distribution as $\frac{1}{2}\left\langle \mb x_{\star},\mb h\right\rangle \left\lVert \mb g_{2}\right\rVert ^{2}+\frac{1}{2}\left\lVert \mb h_{\perp}\right\rVert _{2}\mb g_{1}^{\T}\mb g_{2}$
where $\mb g_{1},\mb g_{2}\in\mbb R^{2}$ are independent standard
Normal random variables: 
\begin{align*}
\P\left(\langle\mb a\mb a^{*}\mb x_{\star},\mb h\rangle>\frac{1}{2}\eta^{-1}\right) & =\P\left(\langle\mb x_{\star},\mb h\rangle\left\lVert \mb g_{2}\right\rVert _{2}^{2}+\left\lVert \mb h_{\perp}\right\rVert _{2}\mb g_{1}^{\T}\mb g_{2}>\eta^{-1}\right).
\end{align*}
 Since $\mb g_{2}$ has a standard normal distribution, its norm
and (normalized) direction are independent. Thus, we can treat $-\mb g_{1}^{\T}\frac{\mb g_{2}}{\left\lVert \mb g_{2}\right\rVert _{2}}=g$
as a standard Normal scalar which is independent of $\left\lVert \mb g_{2}\right\rVert _{2}=v\sim\mr{Rayleigh}(1)$.
Therefore, we have
\begin{align*}
\P\left(\langle\mb a\mb a^{*}\mb x_{\star},\mb h\rangle>\frac{1}{2}\eta^{-1}\right) & =\P\left(\langle\mb x_{\star},\mb h\rangle\left\lVert \mb g_{2}\right\rVert _{2}^{2}+\left\lVert \mb h_{\perp}\right\rVert _{2}\left\lVert \mb g_{2}\right\rVert _{2}\mb g_{1}^{\T}\frac{\mb g_{2}}{\left\lVert \mb g_{2}\right\rVert _{2}}>\eta^{-1}\right)\\
 & =\P\left(\langle\mb x_{\star},\mb h\rangle v-\eta^{-1}v^{-1}>\left\lVert \mb h_{\perp}\right\rVert _{2}g\right).
\end{align*}
 Since $\mb h\in\mc R_{\delta}$ and $\left\lVert \mb h\right\rVert _{2}>\left(t\eta\right)^{-1}$
we have 
\begin{align}
\left(t\eta\right)^{-2} & <\left\lVert \mb h\right\rVert _{2}^{2}=\left\lVert \mb h_{\perp}\right\rVert _{2}^{2}+\left(\mr{Im}\left(\mb x_{\star}^{*}\mb h\right)\right)^{2}+\langle\mb x_{\star},\mb h\rangle^{2}\nonumber \\
 & \le\left(1+\delta^{-2}\right)\left\lVert \mb h_{\perp}\right\rVert _{2}^{2}+\langle\mb x_{\star},\mb h\rangle^{2}.\label{eq:exclusion}
\end{align}
We consider two cases depending on $\left\lVert \mb h_{\perp}\right\rVert _{2}=0$
or not. If $\left\lVert \mb h_{\perp}\right\rVert _{2}=0$, then $\left|\langle\mb x_{\star},\mb h\rangle\right|>\left(t\eta\right)^{-1}$.
The fact that $\mb h\in\mc C'_{\delta}$ as well, implies that $\langle\mb x_{\star},\mb h\rangle$
is non-negative and thereby $\langle\mb x_{\star},\mb h\rangle>\left(t\eta\right)^{-1}$.
Consequently, we have
\begin{align*}
\P\left(\langle\mb a\mb a^{*}\mb x_{\star},\mb h\rangle>\frac{1}{2}\eta^{-1}\right) & =\P\left(\langle\mb x_{\star},\mb h\rangle v-\eta^{-1}v^{-1}>\left\lVert \mb h_{\perp}\right\rVert _{2}g\right)\\
 & \ge\P\left(t^{-1}v-v^{-1}>0\right)=\P(v>\sqrt{t})=e^{-\frac{t}{2}}.
\end{align*}
If $\left\lVert \mb h_{\perp}\right\rVert >0$, then we can invoke
Lemma \ref{lem:rayleigh-normal} in the Appendix with $\alpha=\frac{\langle\mb x_{\star},\mb h\rangle}{\left\lVert \mb h_{\perp}\right\rVert _{2}}$
to show that 
\begin{align*}
\P\left(\langle\mb a\mb a^{*}\mb x_{\star},\mb h\rangle>\frac{1}{2}\eta^{-1}\right) & \ge\P\left(\frac{\langle\mb x_{\star},\mb h\rangle}{\left\lVert \mb h_{\perp}\right\rVert _{2}}v-\frac{\eta^{-1}}{\left\lVert \mb h_{\perp}\right\rVert _{2}}v^{-1}>g\right)\\
 & =\left(\frac{1}{2}+\frac{\alpha}{2\sqrt{\alpha^{2}+1}}\right)\exp\left(-\frac{\eta^{-1}}{\left\lVert \mb h_{\perp}\right\rVert _{2}\sqrt{\alpha^{2}+1}+\alpha}\right).
\end{align*}
 Then by rewriting (\ref{eq:exclusion}) as $\frac{\left(t\eta\right)^{-2}}{\left\lVert \mb h_{\perp}\right\rVert _{2}^{2}}\le1+\delta^{-2}+\alpha^{2}$
we have 
\begin{align*}
\P\left(\langle\mb a\mb a^{*}\mb x_{\star},\mb h\rangle>\frac{1}{2}\eta^{-1}\right) & \ge\left(\frac{1}{2}+\frac{\alpha}{2\sqrt{\alpha^{2}+1}}\right)\exp\left(-\frac{\sqrt{1+\delta^{-2}+\alpha^{2}}}{\sqrt{\alpha^{2}+1}+\alpha}\,t\right).
\end{align*}
The fact that $\mb h\in\mc C'_{\delta}$, guarantees that $\alpha\ge-\sqrt{\delta^{-2}-1}$.
Since $\frac{\alpha}{\sqrt{\alpha^{2}+1}}$ and $-\frac{\sqrt{1+\delta^{-2}+\alpha^{2}}}{\sqrt{\alpha^{2}+1}+\alpha}$
are both increasing in $\alpha$, we obtain 
\begin{align*}
\P\left(\langle\mb a\mb a^{*}\mb x_{\star},\mb h\rangle>\frac{1}{2}\eta^{-1}\right) & \ge\left(\frac{1}{2}-\frac{\sqrt{1-\delta^{2}}}{2}\right)\exp\left(-\frac{\sqrt{2\delta^{-2}}}{\sqrt{\delta^{-2}}-\sqrt{\delta^{-2}-1}}t\right),\\
 & =\left(\frac{1}{2}-\frac{\sqrt{1-\delta^{2}}}{2}\right)\exp\left(-\frac{\sqrt{2}}{1-\sqrt{1-\delta^{2}}}t\right)\\
 & \ge\left(\frac{1}{2}-\frac{\sqrt{1-\delta^{2}}}{2}\right)\exp\left(-2\sqrt{2}\delta^{-2}t\right).
\end{align*}
The above lower bound is the smaller one of the two considered cases
and thus the proof is complete.
\end{proof}

\appendix

\section{\label{sec:SLT}Tools from statistical learning theory}

For reference, here we provide some of the classic results in statistical
learning theory that we employed in our analysis. We mostly follow
the exposition of the subject presented by \citet[chapters 13 and 14]{Devroye-Probabilistic-2013}.
\begin{defn}[Shatter coefficient]
The \emph{$n$-th shatter coefficient }(or\emph{ growth function})
of a class $\mc F$ of binary functions $f:\mc X\to\left\{ 0,1\right\} $
is defined as 
\begin{align*}
s(\mc F,n) & \defeq\max_{\mb x_{1},\mb x_{2}\dotsc,\mb x_{n}\in\mc X}\left|\left\{ \left(f(\mb x_{1}),f(\mb x_{2}),\dotsc,f(\mb x_{n})\right)\,:\,f\in\mc F\right\} \right|.
\end{align*}
\end{defn}
Intuitively, the shatter coefficient $s(\mc F,n)$ is the largest
number of binary patterns that the functions in $\mc F$ can induce
on $n$ points.
\begin{defn}[VC\textendash dimension]
The \emph{Vapnik\textendash Chervonenkis (VC) dimension} of a class
$\mc F$ of binary functions is the largest number $n$ such that
$s(\mc F,n)=2^{n}$, namely,
\[
\dim_{\mr{VC}}(\mc F)\defeq\max\left\{ n\,:\,s(\mc F,n)=2^{n}\right\} .
\]
Naturally, $\dim_{\mr{VC}}(\mc F)=\infty$ if $s(\mc F,n)=2^{n}$
for all $n$.
\end{defn}
If $\mc F$ can induce all binary patterns on $n$ points, $\mc F$
is said to ``shatter'' $n$ points. Therefore, the VC\textendash dimension
of $\mc F$ is the largest number of points that $\mc F$ can shatter.
\begin{lem}[\citet{Vapnik-Uniform-1971,Sauer-Density-1972,Shelah-Combinatorial-1972}]
\label{lem:Sauer's}For a class $\mc F$ of binary functions with
VC\textendash dimension $d=\dim_{\mr{VC}}(\mc F)$ we have 
\[
s\left(\mc F,n\right)\le\sum_{i=0}^{d}\binom{n}{i}\,.
\]
In particular,
\begin{align}
s(\mc F,n) & \le\left(\frac{en}{d}\right)^{d}\,.\label{eq:shatter-coeff-bound}
\end{align}
\end{lem}
The following theorem is originally due to \citet{Vapnik-Uniform-1971}.
We restate the theorem as presented in \citet{Devroye-Probabilistic-2013}.
\begin{thm}[\citet{Vapnik-Uniform-1971}]
\label{thm:VC}Let $\mc F$ be a class of binary functions and $\mb x_{1},\mb x_{2},\dots,\mb x_{n}$
be i.i.d. copies of an arbitrary random variable $\mb x$. Then for
every $t>0$ we have 
\[
\P\left(\sup_{f\in\mc F}\left|\frac{1}{n}\sum_{i=1}^{n}f(\mb x_{i})-\E f(\mb x)\right|>t\right)\le8s(\mc F,n)e^{-nt^{2}/8}.
\]
\end{thm}

\section{Auxiliary Lemma}

\begin{lem}
\label{lem:rayleigh-normal}Let $v\sim\mr{Rayleigh}(1)$ and $g\sim\mr{Normal}(0,1)$
be independent random variables. Then we have 
\begin{align*}
\P(\alpha v+\beta v^{-1}>g) & =\begin{cases}
1-\frac{\sqrt{\alpha^{2}+1}-\alpha}{2\sqrt{\alpha^{2}+1}}e^{-\beta\left(\alpha+\sqrt{\alpha^{2}+1}\right)} & \text{for }\text{\ensuremath{\beta\ge}0}\\
\frac{\sqrt{\alpha^{2}+1}+\alpha}{2\sqrt{\alpha^{2}+1}}e^{\beta/\left(\alpha+\sqrt{\alpha^{2}+1}\right)} & \text{for }\beta<0.
\end{cases}
\end{align*}
for all $\alpha,\beta\in\mbb R$.
\end{lem}
\begin{proof}
We denote the standard normal cumulative distribution function and
its derivative by $\Phi(\cdot)$ and $\phi(\cdot)$, respectively.
Let $F(\beta)=\P(\alpha v+\beta v^{-1}>\gamma)=\E\Phi\left(\alpha v+\beta v^{-1}\right)$.
By Leibniz's rule we have 
\begin{align*}
F'(\beta) & =\E\left(\phi\left(\alpha v+\beta v^{-1}\right)v^{-1}\right)\\
 & =\frac{1}{\sqrt{2\pi}}\int_{0}^{\infty}e^{-\frac{1}{2}\left(\alpha v+\beta v^{-1}\right)^{2}}v^{-1}\cdot ve^{-\frac{1}{2}v^{2}}\d v\\
 & =\frac{e^{-\alpha\beta}}{\sqrt{2\pi}}\int_{0}^{\infty}e^{-\frac{1}{2}\left(\left(\alpha^{2}+1\right)v^{2}+\beta^{2}v^{-2}\right)}\d v.
\end{align*}
Now let $G(\beta)=\frac{1}{\sqrt{2\pi}}\int_{0}^{\infty}e^{-\frac{1}{2}\left(\left(\alpha^{2}+1\right)v^{2}+\beta^{2}v^{-2}\right)}\d v,$
so that $F'(\beta)=e^{-\alpha\beta}G(\beta)$. Using Leibniz's rule
again, we can write 
\begin{align*}
G'(\beta) & =-\frac{\beta}{\sqrt{2\pi}}\int_{0}^{\infty}v^{-2}e^{-\frac{1}{2}\left(\left(\alpha^{2}+1\right)v^{2}+\beta^{2}v^{-2}\right)}\d v\\
 & =-\frac{\sqrt{\alpha^{2}+1}}{\sqrt{2\pi}}\int_{0}^{\infty}e^{-\frac{1}{2}\left(\left(\alpha^{2}+1\right)u^{2}+\beta^{2}u^{-2}\right)}\d u\\
 & =-\sqrt{\alpha^{2}+1}\,G(\beta),
\end{align*}
 where the second line follows from the change of variable $v=\frac{\beta}{\sqrt{\alpha^{2}+1}}u^{-1}$.
It is straightforward to show that $G(0)=\frac{1}{2\sqrt{\alpha^{2}+1}}$.
A simple integration then yields
\[
G(\beta)=G(0)e^{-\beta\sqrt{\alpha^{2}+1}}=\frac{1}{2\sqrt{\alpha^{2}+1}}e^{-\beta\sqrt{\alpha^{2}+1}}
\]
 for $\beta\ge0$, and since $G(\beta)$ is even for all $\beta$
we have 
\[
G(\beta)=\frac{1}{2\sqrt{\alpha^{2}+1}}e^{-\left|\beta\right|\sqrt{\alpha^{2}+1}}.
\]
 It then follows that
\begin{align*}
F'(\beta) & =\frac{1}{2\sqrt{\alpha^{2}+1}}e^{-\left(\alpha\beta+\left|\beta\right|\sqrt{\alpha^{2}+1}\right)}
\end{align*}
 Integrating again we obtain 
\begin{align*}
F(\beta) & =\begin{cases}
F(0)+\frac{\sqrt{\alpha^{2}+1}-\alpha}{2\sqrt{\alpha^{2}+1}}\left(1-e^{-\beta\left(\alpha+\sqrt{\alpha^{2}+1}\right)}\right) & ,\text{\ensuremath{\beta\ge}0}\\
F(0)-\frac{\sqrt{\alpha^{2}+1}+\alpha}{2\sqrt{\alpha^{2}+1}}\left(1-e^{-\beta\left(\alpha-\sqrt{\alpha^{2}+1}\right)}\right) & ,\beta<0.
\end{cases}
\end{align*}
We can calculate $F(0)$ as 
\begin{align*}
F(0) & =\P\left(\alpha v>g\right)\\
 & =\begin{cases}
\frac{1}{2}+\frac{1}{2}\P(\alpha^{2}v^{2}>g^{2}) & \text{for }\alpha\ge0\\
\frac{1}{2}\P(\alpha^{2}v^{2}<g^{2}) & \text{for }\alpha<0
\end{cases}\\
 & =\begin{cases}
\frac{1}{2}+\frac{1}{2}\P(\frac{\alpha^{2}}{\alpha^{2}+1}>\frac{g^{2}}{v^{2}+g^{2}}) & \text{for }\alpha\ge0\\
\frac{1}{2}-\frac{1}{2}\P(\frac{\alpha^{2}}{\alpha^{2}+1}>\frac{g^{2}}{v^{2}+g^{2}}) & \text{for }\alpha<0
\end{cases}\\
 & =\frac{\sqrt{\alpha^{2}+1}+\alpha}{2\sqrt{\alpha^{2}+1}},
\end{align*}
 where the last line follows from the fact that $\frac{g}{\sqrt{v^{2}+g^{2}}}$
has a uniform distribution over $[-1,1]$. Replacing $F(0)$ in the
expression of $F(\beta)$ and straightforward simplifications yield
the desired result.
\end{proof}


\bibliographystyle{abbrvnat}
\bibliography{references}

\end{document}